\documentclass[runningheads]{llncs}

\usepackage{ltl}
\usepackage{graphicx}
\usepackage{float}
\usepackage{xspace}
\usepackage{mathtools}
\usepackage{ltl}
\usepackage{scalefnt}
\usepackage{cleveref}
\usepackage{todo}
\usepackage{caption}
\usepackage{multirow}
\usepackage{ stmaryrd }
\usepackage{cite}  
\usepackage[firstpage]{draftwatermark} 

\usepackage{subcaption}
\let\counterwithin\relax
\usepackage{chngcntr}
\counterwithin*{equation}{section}

\usepackage{xcolor}

\usepackage{tikz}
\usetikzlibrary
{ chains
	, calc
	, arrows
	, positioning
	, decorations
	, shapes
	, automata
	, external
	, shadows
	, fadings
	, decorations.pathreplacing
	, decorations.pathmorphing
	, patterns
	, shapes.callouts
	, shapes.arrows
	, backgrounds
	, matrix
	, scopes
}

\newcommand{\penv}{{p_\textit{env}}}

\newcommand{\U}{\LTLuntil}
\newcommand{\X}{\LTLnext}
\newcommand{\G}{\LTLglobally}
\newcommand{\F}{\LTLeventually}
\newcommand{\R}{\LTLrelease}

\newcommand{\K}[1]{\mathop{\mathcal{K}_{#1}}}

\newcommand{\nat}{\mathbb{N}}
\newcommand{\pow}[1]{2^{#1}}
\newcommand{\cupdot}{\mathbin{\dot\cup}}
\newcommand{\set}[1]{\{ #1 \}}

\newcommand{\strat}[2]{(2^{#1})^* \rightarrow 2^{#2}}

\newcommand{\ap}{\text{AP}}
\newcommand{\freeap}{\text{AP'}}

\newcommand{\pathvars}{\mathcal{V}}
\newcommand{\pathassign}{\Pi}
\newcommand{\ldot}{\mathpunct{.}}

\newcommand{\dep}[2]{D_{#1 \mapsto #2}}
\newcommand{\collapse}{\mathit{collapse}}

\newcommand{\LTLK}{$\text{LTL}_\mathcal{K}$\xspace}
\newcommand{\HQPTLP}{HyperQPTL\hspace{-1pt}\textsuperscript{\tiny\textbf +}\xspace}

\newcommand{\zthree}{\textsc{z3}}

\definecolor{green}{RGB}{48, 107, 52}
\definecolor{red}{RGB}{181,23,23}

\begin{document}
\title{Realizing $\omega$-regular Hyperproperties\thanks{This work was partially supported by the Collaborative Research Center ``Foundations of Perspicuous Software Systems'' (TRR 248, 389792660) and by the European Research Council (ERC) Grant OSARES (No. 683300).}}
%
%
\author{Bernd Finkbeiner\orcidID{0000-0002-4280-8441}\and
Christopher Hahn\orcidID{0000-0002-1243-4880}\and
Jana Hofmann\orcidID{0000-0003-1660-2949}\and
Leander Tentrup\orcidID{0000-0002-6150-2982}}


\institute{Reactive Systems Group\\ Saarland University \\
\email{\{finkbeiner, hahn, hofmann, tentrup\}@react.uni-saarland.de}}
\SetWatermarkText{\hspace*{5.5in}\raisebox{5.5in}{\includegraphics[scale=0.1]{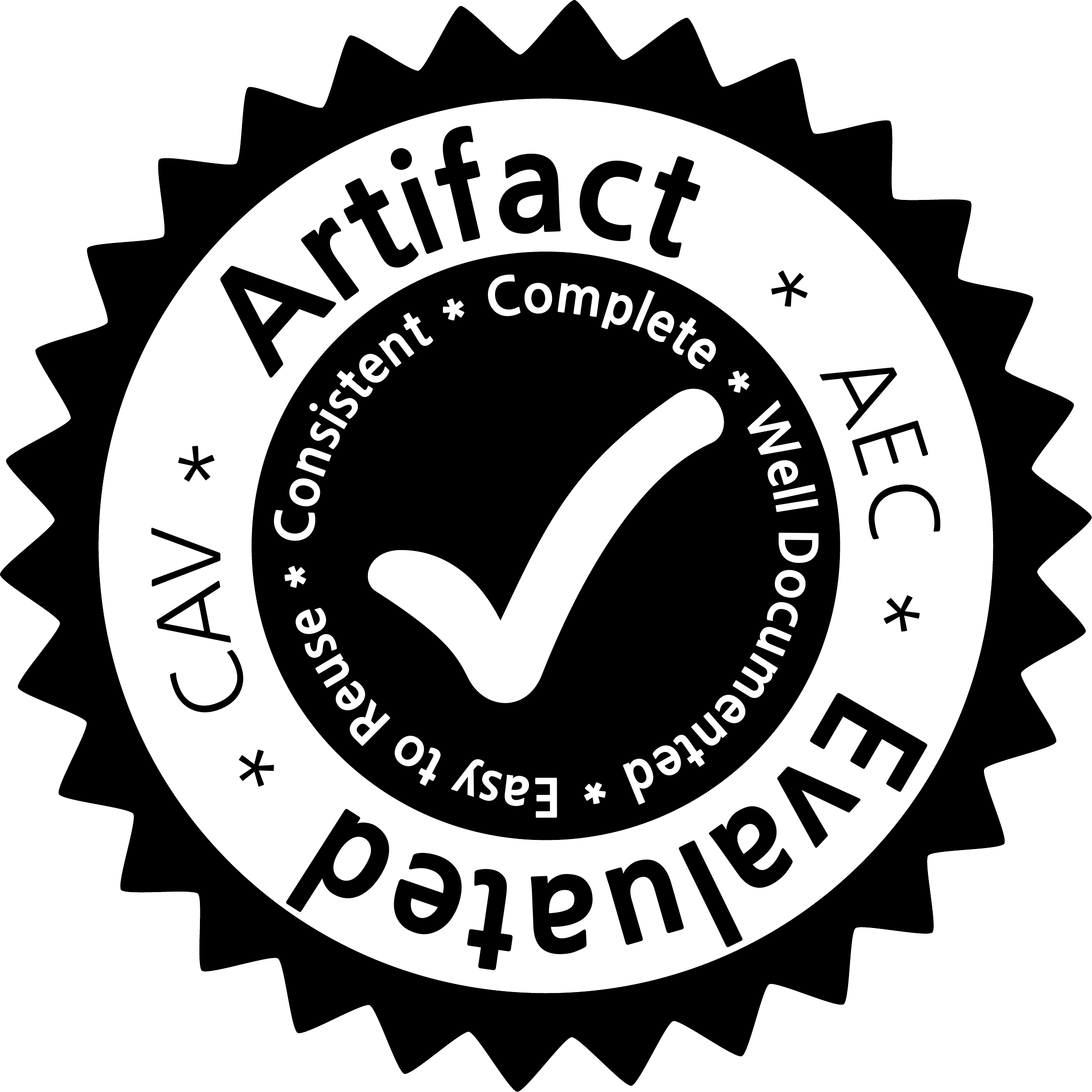}}}
\SetWatermarkAngle{0}
\maketitle              
\begin{abstract}
We study the expressiveness and reactive synthesis problem of HyperQPTL, a logic that specifies $\omega$-regular hyperproperties.
HyperQPTL is an extension of linear-time temporal logic (LTL) with explicit trace and propositional quantification and therefore \emph{truly} combines trace relations and $\omega$-regularity.
As such, HyperQPTL can express promptness, which states that there is a common bound on the number of steps up to which an event must have happened.
We demonstrate how the HyperQPTL formulation of promptness differs from the type of promptness expressible in the logic Prompt-LTL.
Furthermore, we study the realizability problem of HyperQPTL by identifying decidable fragments, where one decidable fragment contains formulas for promptness. We show that, in contrast to the satisfiability problem of HyperQPTL, propositional quantification has an immediate impact on the decidability of the realizability problem.
We present a reduction to the realizability problem of HyperLTL, which immediately yields a bounded synthesis procedure.
We implemented the synthesis procedure for HyperQPTL in the bounded synthesis tool BoSy. Our experimental results show that a range of arbiter satisfying promptness can be synthesized.

\end{abstract}

\section{Introduction}
\label{sec:intro}
	Hyperproperties~\cite{journals/jcs/ClarksonS10}, which are mainly studied in the area of secure information flow control, are a generalization from trace properties to \emph{sets} of trace properties. That is, they relate multiple execution traces with each other. 
	Examples are noninterference~\cite{conf/sp/GoguenM82a}, observational determinism~\cite{conf/csfw/ZdancewicM03}, symmetry~\cite{conf/cav/FinkbeinerRS15}, or promptness~\cite{journals/fmsd/KupfermanPV09}, i.e., properties whose satisfaction cannot be determined by analyzing each execution trace in isolation.
	
	A number of logics have been introduced to express hyperproperties (examples are~\cite{conf/post/ClarksonFKMRS14, conf/stacs/Finkbeiner017, conf/memocode/NguyenKJDJ17}). They either add explicit trace quantification to a temporal logic or build on monadic first-order or second-order logics and add an equal-level predicate, which connects traces with each other.	
	A comprehensive study comparing such hyperlogics has been initiated in~\cite{hierarchy_hyperlogics}.
	
	The most prominent hyperlogic is HyperLTL~\cite{conf/post/ClarksonFKMRS14}, which extends classic linear-time temporal logic (LTL)~\cite{conf/focs/Pnueli77} with trace variables and explicit trace quantification.
	HyperLTL has been successfully applied in (runtime) verification, (e. g.,~\cite{DBLP:conf/fm/StuckiSSB19, DBLP:conf/cav/FinkbeinerHT18,DBLP:conf/rv/Hahn19}), specification analysis~\cite{DBLP:conf/cav/FinkbeinerHS17,DBLP:conf/atva/FinkbeinerHH18}, synthesis~\cite{conf/cav/FinkbeinerHLST18,DBLP:journals/acta/FinkbeinerHLST20}, and program repair~\cite{DBLP:conf/atva/BonakdarpourF19} of hyperproperties.
	As an example specification, the following HyperLTL formula expresses observational determinism by stating that for every pair of traces, if the observable inputs $I$ are the same on both traces, then also the observable outputs $O$ have to agree
	\begin{equation}
	\forall \pi \forall \pi' \ldot \G (I_\pi = I_{\pi'}) \rightarrow \G (O_\pi = O_{\pi'}) \enspace.
	\end{equation}
	
	Thus, hyperlogics can not only specify functional correctness, but may also enforce the absence of information leaks or presence of information propagation. 
	There is a great practical interest in information flow control, which makes synthesizing implementations that satisfy hyperproperties highly desirable.
	Recently~\cite{conf/cav/FinkbeinerHLST18}, it was shown that the synthesis problem of HyperLTL, although undecidable in general, remains decidable for many fragments, such as the $\exists^*\forall$ fragment. Furthermore, a \emph{bounded synthesis} procedure was developed, for which a prototype implementation based on BoSy\cite{conf/cav/FaymonvilleFT17,conf/cav/FinkbeinerHLST18,conf/cav/CoenenFST19} showed promising results.

	HyperLTL is, however, intrinsically limited in expressiveness. 
	For example, promptness is not expressible in HyperLTL. Promptness is a property stating that there is a bound $b$, common for all traces, on the number of steps up to which an event $e$ must have happened.
	Additionally, just like LTL, HyperLTL can express neither $\omega$-regular nor epistemic properties~\cite{MarkusThesis, conf/fossacs/BozzelliMP15}.
	Epistemic properties are statements about the transfer of knowledge between several components. An exemplary epistemic specification is described by the \emph{dining cryptographers problem}~\cite{journals/cacm/Chaum85}: three cryptographers sit at a table in a restaurant. Either one of the cryptographers or, alternatively, the NSA must pay for their meal.
	The question is whether there is a protocol where each cryptographer can find out whether the NSA or one of the cryptographers paid the bill, without revealing the identity of the paying cryptographer.
	
	In this paper, we explore HyperQPTL~\cite{MarkusThesis, hierarchy_hyperlogics}, a hyperlogic that is more expressive than HyperLTL. Specifically, we study its expressiveness and reactive synthesis problem. HyperQPTL extends HyperLTL with quantification over sequences of new propositions. What makes the logic particularly expressive is the fact that the trace quantifiers and propositional quantifiers can be freely interleaved. With this mechanism, HyperQPTL can not only express all $\omega$-regular properties over a sequences of n-tuples; it truly interweaves trace quantification and $\omega$-regularity.
	For example, promptness can be stated as the following HyperQPTL formula:
	\begin{equation}
	\label{form:promptness}
	\exists b. \forall \pi.~ \F b \land (\neg b ~ \U e_\pi) \enspace .	
	\end{equation}
	The formula states that there exists a sequence $s \in (\pow{\set{q}})^\omega$, such that event $e$ holds on all traces before the first occurrence of $b$ in $s$. In this paper, we argue that the type of promptness expressible in HyperQPTL is incomparable to the expressiveness of Prompt-LTL~\cite{journals/fmsd/KupfermanPV09}, a logic introduced to express promptness properties. It is further known that HyperQPTL also subsumes epistemic extensions of temporal logics such as \LTLK~\cite{journals/jcss/HalpernV89}, as well as the first-order hyperlogic FO[$<, E$]~\cite{conf/stacs/Finkbeiner017, MarkusThesis, hierarchy_hyperlogics}. Its expressiveness makes HyperQPTL particularly interesting.
	The model checking problem of HyperQPTL is, despite the logic being quite expressive, decidable~\cite{MarkusThesis}. We also explore an alternative definition of HyperQPTL that would result in an even more expressive logic. However, we show that the logic would have an undecidable model checking problem, which constitutes a major drawback in the context of computer-aided verification.
	Furthermore, satisfiability is decidable for large fragments of the logic~\cite{hierarchy_hyperlogics}.
	Decidable HyperQPTL fragments can be described solely in terms of their \emph{trace} quantifier prefix.
	This indicates that propositional quantification has no negative impact on the decidability, although it greatly increases the expressiveness.
	We establish that propositional quantification, in contrast to the satisfiability problem, has an impact on the realizability problem:
	it becomes undecidable when combining a propositional $\forall \exists$ quantifier alternation with a single universal trace quantifier.
	However, we show that the synthesis problem of large HyperQPTL fragments remains decidable, where one of these fragments contains promptness properties.
	We partially obtain these results by reducing the HyperQPTL realizability problem to the HyperLTL realizability problem. Based on this reduction, we extended the BoSy bounded synthesis tool to also synthesize systems respecting HyperQPTL specifications.
	We provide promising experimental results of our prototype implementation: using BoSy and HyperQPTL specifications, we were able to synthesize arbiters that respect promptness.
	
		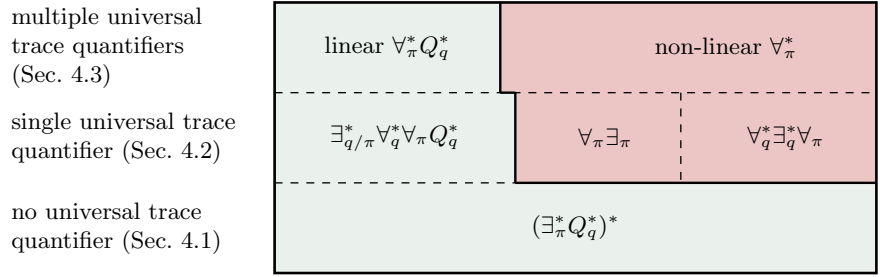
\begin{figure}[t]
		\centering
		\pgfdeclarelayer{background}
		\pgfdeclarelayer{foreground}
		\pgfsetlayers{background,main,foreground}
		\begin{tikzpicture}
		\filldraw[green!10] (4,1.8) rectangle (-4,-1.8) {};
		\filldraw[red!25] (4,1.8) rectangle (-1,.6) {};
		\filldraw[red!25] (4,1.8) rectangle (-.8,-.6) {};
		
		\node[draw=black, rectangle, line width=1pt, color=black, minimum width=8cm,minimum height=3.6cm] at (0,0) {};
		
		\draw[dashed, line width=.5pt] (-4,-.6) -- (4,-.6);
		\draw[dashed, line width=.5pt] (-4,0.6) -- (4,.6);
		\draw[dashed, line width=.5pt] (1.4, 0.6) -- (1.4, -0.6);
		
		\draw[line width=.9pt] (-1, 1.8) -- (-1, 0.6);
		\draw[line width=.9pt] (-.8, 0.6) -- (-.8, -0.6);
		\draw[line width=.9pt] (-.8,-.6) -- (4,-.6);
		\draw[line width=.9pt] (-1,.6) -- (-.8,.6);
		
		\node at (0, -1.2) {$(\exists_{\pi}^* Q_q^*)^*$};
		
		\node at (-2.4, 0) {$\exists_{q/\pi}^*\forall_q^* \forall_{\pi} Q_q^*$};
		\node at (.4, 0) {$\forall_\pi \exists_\pi$};
		\node at (2.8, 0) {$\forall_q^*\exists_q^*\forall_\pi$};
		
		\node at (-2.5, 1.2) {linear $\forall_{\pi}^*Q_q^*$};		
		\node at (2, 1.2) {non-linear $\forall_{\pi}^*$};
		
		\node[text width=3cm] at (-6, 1.2) {multiple universal trace quantifiers (Sec.~\ref{sec:MultUnivTrace})};
		\node[text width=3cm] at (-6, 0) {single universal trace quantifier (Sec.~\ref{sec:SingleUnivTrace})};
		\node[text width=3cm] at (-6, -1.2) {no universal trace quantifier (Sec.~\ref{sec:NoUnivTrace})};

		\end{tikzpicture}
		\caption{The realizability problem of HyperQPTL. Left and below of the solid line are the decidable fragments, right above the solid line the undecidable fragments.}
		\label{fig:fragments}
	\end{figure}
	
	This paper is structured as follows.
	In Section~\ref{sec:prelims}, we give necessary preliminaries.
	In Section~\ref{sec:omega-regular-hyperproperties}, we define HyperQPTL. We discuss an alternative approach to define a logic expressing $\omega$-regular hyperproperties, before pointing out that its model checking problem is undecidable. Subsequently, we give examples for the expressiveness of HyperQPTL, namely by characterizing the type of promptness properties HyperQPTL can express. Additionally, we recapitulate how HyperQPTL also subsumes epistemic properties.
	Section~\ref{sec:HyperQPTL_realizability} discusses the realizability problem of HyperQPTL. We describe HyperQPTL fragments in terms of their quantifier prefixes. To present our results, we use the following notation. We write $\forall_\pi$ and $\forall_q$ for a single universal trace and propositional quantifier, respectively. To denote a sequence of universal trace and propositional quantifiers, we write $\forall_\pi^*$ and $\forall_q^*$. Furthermore, we use $\forall_{\pi/q}^*$ for a sequence of mixed universal quantification. We use the analogous notation for existential quantifiers. Lastly, $Q_\pi^*$ and $Q_q^*$ denote a sequence of mixed universal and existential trace and propositional quantifiers, respectively. As an example, the $\forall_\pi^* Q_q^*$ fragment denotes all formulas of the form $\forall \pi_1 \ldot \ldots \forall \pi_m \ldot \exists/\forall q_1 \ldot \ldots \exists/\forall q_n \ldot \varphi$, where $\varphi$ is quantifier free.
	Figure~\ref{fig:fragments} summarizes our results. We establish that a major factor for the decidability of the realizability problem consists in the number of universal trace occurring in a formula. Realizability of HyperQPTL formulas without $\forall \pi$ quantifiers is decidable (Section~\ref{sec:NoUnivTrace}). Formulas with a single $\forall \pi$ are decidable if they belong to the $\exists_{q/\pi}^*\forall_q^* \forall_{\pi} Q_q^*$ fragment. This fragment also contains promptness. For more than one universal trace quantifier, we show that decidability can be guaranteed for a fragment that we call the linear $\forall_{\pi}^*Q_q^*$ fragment. We also show that all the above fragments are tight, i.e., realizability of all other formulas is in general undecidable.
	Lastly, Section~\ref{sec:Experiments} presents experiments for the prototype implementation of our bounded synthesis algorithm for HyperQPTL.

\section{Preliminaries}
\label{sec:prelims}

We use $\ap$ for a set of atomic propositions. A \emph{trace} over $\ap$ is an infinite sequence $t \in (\pow{\ap})^\omega$. For $i \in \nat$, we write $t[i]$ for the $i$th element of $t$ and $t[i, \infty]$ for the suffix of t starting from position $i$.
For two traces $t, t'$ over $\ap$ and a set $\freeap \subseteq \ap$, we write $t =_\freeap t'$ to indicate that $t$ and $t'$ agree on all $a \in \freeap$, and respectively $T =_\freeap T'$ for two sets of traces $T$ and $T'$.
Furthermore, we define a replacement function $t [q \mapsto t_q]$ that given a trace $t$ and a trace $t_q \in (\pow{\{q\}})^\omega$, replaces the occurrences of $q$ in $t$ according to $t_q$, such that $t[q \mapsto t_q] =_{\{q\}} t_q$ and $t[q \mapsto t_q] =_{\ap \backslash \set{q}} t$. We also lift this notation to sets of traces and define $T[q \mapsto t_q] = \set{t[q \mapsto t_q]  \mid t \in T}$.

	QPTL~\cite{QPTL} extends Linear Temporal Logic (LTL) with quantification over propositions. QPTL formulas $\varphi$ are defined as follows.
	\begin{align*}
	\varphi &\Coloneqq \exists q \ldot \varphi \mid \forall q \ldot \varphi \mid \psi \\
	\psi &\Coloneqq q \mid \neg \psi \mid \psi \lor \psi \mid \X \psi \mid \F \psi
	\end{align*}
	where $q \in \ap$ and $\ap$ is a set of atomic propositions.
	For simplicity, we assume that variable names in formulas are cleared of double occurrences.
	The semantics of $\varphi$ over $\ap$ is defined with respect to a trace $t \in (\pow{\ap})^\omega$.
	\begin{alignat*}{3}
		t &\models q                            &&\text{ iff } && q \in t[0] \\
		t &\models \neg \psi                  &&\text{ iff } && t \not\models  \psi \\
		t &\models \psi_1 \vee \psi_2      \quad &&\text{ iff } \quad &&t \models  \psi_1 \text{ or } t \models  \psi_2 \\
		t &\models \LTLnext \psi            &&\text{ iff } && t[1,\infty] \models  \psi \\
		t &\models \F \psi &&\text{ iff } && \exists i \geq 0.~ t[i,\infty] \models  \psi \\
		t &\models \exists q \ldot \varphi &&\text{ iff } && \exists t_q \in  (\pow{\{q\}} )^\omega \ldot t[q \mapsto t_q] \models \varphi  \\
		t &\models \forall q \ldot \varphi &&\text{ iff } && \forall t_q \in  (\pow{\{q\}} )^\omega \ldot t[q \mapsto t_q] \models \varphi  
	\end{alignat*}
	We did not define the until operator $\LTLuntil$ as native part of the logic. It can be derived using propositional quantification~\cite{kaivola1997QPTLUntil}. The boolean connectives $\land, \rightarrow, \leftrightarrow$ and the temporal operators globally $\G$ and release $\LTLrelease$ are derived as usually.

\section{$\omega$-Regular Hyperproperties}
\label{sec:omega-regular-hyperproperties}
	Just like LTL, HyperLTL cannot express $\omega$-regular languages~\cite{MarkusThesis}. LTL can be extended to QPTL by adding quantification over atomic propositions. In QPTL, $\omega$-regular languages become expressible. We therefore study HyperQPTL~\cite{MarkusThesis, hierarchy_hyperlogics}, the extension of HyperLTL with propositional quantification, to express $\omega$-regular hyperproperties. 
	Given a set $\ap$ of atomic propositions and a set $\pathvars$ of trace variables, the syntax of HyperQPTL is defined as follows
	\begin{align*}
	\varphi &{}\Coloneqq \forall\pi\ldot\varphi \mid \exists\pi\ldot\varphi \mid \forall q \ldot\varphi \mid \exists q \ldot\varphi \mid\psi \enspace \\
	\psi &{}\Coloneqq a_\pi \mid q \mid \neg\psi \mid \psi\lor\psi \mid \X\psi \mid \F\psi \enspace ,
	\end{align*}
	where $a, q \in \ap$ and $\pi \in \pathvars$.
	As for QPTL, we assume that formulas are cleared of double occurrences of variable names.
	We require that in well-defined HyperQPTL formulas, each $a_\pi$ is in the scope of a trace quantifier binding $\pi$ and each $q$ is in the scope of a propositional quantifier binding $q$.
	Note that atomic propositions $a_\pi$ refer to a quantified trace $\pi$, whereas quantified propositional variables $q$ are independent of the traces.
	The semantics of a well-defined HyperQPTL formula over $\ap$ is defined with respect to a set of traces $T \subseteq (2^\ap)^\omega$ and an assignment function $\pathassign : \pathvars \to T$. We define the satisfaction relation $\Pi, i \models_T \varphi$ as follows:
	\begin{alignat*}{3}
	&\pathassign,i \models_T a_\pi       &&\text{ iff } &&a \in \pathassign(\pi)[i] \\
	&\pathassign,i \models_T q       &&\text{ iff } && \forall t \in T \ldot q \in t[i] \\
	&\pathassign,i \models_T \neg \psi              && \text{ iff } &&\pathassign,i \not\models_T \psi \\
	&\pathassign,i \models_T \psi_1 \lor \psi_2        \quad && \text{ iff } \quad &&\pathassign,i \models_T \psi_1 \lor \pathassign,i \models_T \psi_2 \\
	&\pathassign,i \models_T \X \psi                && \text{ iff } &&\pathassign,i+1 \models_T \psi \\
	&\pathassign,i \models_T \F \psi             && \text{ iff } &&\exists j \geq i \ldot ~ \pathassign \ldot j \models_T \psi \\
	&\pathassign,i \models_T \exists \pi \ldot \varphi && \text{ iff } && \exists t \in T \ldot \pathassign[\pi \mapsto t], i \models_T \varphi\\
	&\pathassign,i \models_T \forall \pi \ldot \varphi && \text{ iff } && \forall t \in T \ldot \pathassign[\pi \mapsto t], i \models_T \varphi \\
	&\pathassign,i \models_T \exists q \ldot \varphi && \text{ iff } && \exists t_q \in (2^{\{q\}})^\omega \ldot \pathassign, i \models_{T[q \mapsto t_q]} \varphi\\
	&\pathassign,i \models_T \forall q \ldot \varphi && \text{ iff } && \forall t_q \in (2^{\{q\}})^\omega \ldot \pathassign, i \models_{T[q \mapsto t_q]} \varphi \enspace .
	\end{alignat*}
	Note that the semantics of propositional quantification is defined in such a way that in the scope of a quantifier binding $q$, all traces agree on their $q$-sequence.
	We say that a set of traces $T$ satisfies a HyperQPTL formula $\varphi$ if $\emptyset, 0 \models_T \varphi$, where $\emptyset$ is the empty trace assignment.
	QPTL formulas can be expressed in HyperQPTL using a single universal trace quantifier. Furthermore, HyperLTL~\cite{conf/post/ClarksonFKMRS14} is the syntactic subset of HyperQPTL that does not contain propositional quantification.
	
	While HyperQPTL can express a wide range of properties (see Section~\ref{sec:Expressiveness_HyperQPTL}), its model checking problem is still decidable~\cite{MarkusThesis}. Furthermore, the syntactic fragments for which satisfiability is decidable can be expressed solely in terms of the occurring trace quantifiers: Just like for HyperLTL, satisfiability of a HyperQPTL formula is decidable if no $\forall \pi$ is followed by an $\exists \pi$~\cite{hierarchy_hyperlogics}.
	
	The definition of HyperQPTL is straightforward, however, one could argue that it is not the only way to extend QPTL to a hyperlogic. The original idea of QPTL is to ``color'' the trace by introducing additional atomic propositions. The way HyperQPTL is defined, that idea is translated to sets of traces by coloring the traces uniformly. An alternative approach could be to color every trace individually by introducing a full atomic proposition for every propositional quantification. This resembles full second-order quantification and would therefore result in a considerably more expressive logic. In particular, we show that the model checking problem would become undecidable, which is, especially in the context of automatic verification, unfavorable. For the remainder of this section, we call the logic resulting from the alternative definition \HQPTLP.
	The syntax of \HQPTLP is similar to the one of HyperQPTL, just without the rule $q$ for the evaluation of the propositional variables. This accounts for the idea that the propositional quantification can freely reassign atomic propositions; thus, there is no need to distinguish between free atomic propositions and quantified atomic propositions: 
		\begin{align*}
			\varphi &{}\Coloneqq \forall\pi\ldot\varphi \mid \exists\pi\ldot\varphi \mid \forall a \ldot\varphi \mid \exists a \ldot\varphi \mid\psi \enspace \\
			\psi &{}\Coloneqq a_\pi \mid \neg\psi \mid \psi\lor\psi \mid \X\psi \mid \F\psi \enspace .
		\end{align*}
		Semantically, only the rules for the quantification of the propositional quantifiers change:
		\begin{alignat*}{3}
		&\pathassign,i \models_T \exists a \ldot \varphi \quad && \text{ iff } \quad && \exists T' \subseteq (\pow{\ap})^\omega \ldot T' =_{\ap \backslash \{a\}} T \land \pathassign, i \models_{T'} \varphi\\
		&\pathassign,i \models_T \forall a \ldot \varphi && \text{ iff } && \forall T' \subseteq (\pow{\ap})^\omega \ldot T' =_{\ap \backslash \{a\}} T \rightarrow \pathassign, i \models_{T'} \varphi \enspace .
		\end{alignat*}
	\begin{lemma}
		The \HQPTLP model checking problem is undecidable.
	\end{lemma}
	\begin{proof}
		Given a finite Kripke structure $K$ and a \HQPTLP formula $\varphi$, the model checking problem asks whether the trace set $T$ produced by $K$ satisfies $\varphi$. The proof follows the undecidability proof for the model checking problem of S1S[$E$]~\cite{hierarchy_hyperlogics}, a logic which lifts S1S to the level of hyperlogics. We describe a reduction from the halting problem of 2-counter machines (which are Turing complete) to the \HQPTLP model checking problem. A 2-counter machine (2CM) consists of a finite set of serially numbered instructions that modify two counters. A configuration of a 2CM is a triple $(n, v_1, v_2) \in \nat^3$, where $n$ determines the next instruction to be executed, and $v_1$ and $v_2$ assign the counter values. Each instruction can either increase or decrease one of the counters; or test either of the counters for zero and, depending on the outcome, jump to another instruction.
		Furthermore, we assume a special instruction $i_\mathit{halt}$, which indicates that the machine has reached a halting state.  A 2CM halts from initial configuration $s_0$ if there is a finite sequence $s_0, \ldots , s_n$ of configurations such that $s_n$ is a halting configuration and $s_{i+1}$ is a result of applying the instruction in $s_i$ to configuration $s_i$. 
		Let $\mathcal{M}$ be a 2CM. We describe $T$ and $\varphi$ such that $T \models \varphi$ iff $\mathcal{M}$ halts. We choose $\ap = \set{i, c_1, c_2}$ and $T$ is the set of all traces where each atomic proposition holds exactly once. That way, a trace $t$ encodes a configuration of the machine: If $i \in t[n]$, $c_1 \in t[v_1]$, and $c_2 \in t[v_2]$, the machine is in configuration $(n, v_1, v_2)$. It is easy to see that $T$ can be produced by a finite Kripke structure. To describe $\varphi$, we make two helpful observations. First, using propositional quantification, we can quantify a trace set $T_q \subseteq T$: a trace $t$ is in $T_q$ iff the quantified proposition $q$ eventually occurs on $t$. Second, for two traces $t, t' \in T$, we can state that $t'$ encodes a configuration which is the successor of the configuration encoded by $t$.
		Using these observations, we define $\varphi = \exists q \ldot \varphi'$, where $q$ encodes a set $T_q \subseteq T$ that is supposed to describe a halting computation. To ensure that $T_q$ describes a halting computation, $\varphi'$ is a conjunction of the following requirements: $T_q$ must
		\begin{enumerate}
			\item be finite, 
			\item contain a halting configuration and the initial configuration,
			\item be predecessor closed with respect to the encoded configurations it contains (except for the initial configuration).
		\end{enumerate}
		Finiteness of $T_q$ can be expressed by stating that there is an upper bound on the values of $i, c_1$, and $c_2$ on the traces in $T_q$. With the observations made before, stating the above requirements in \HQPTLP now remains a straightforward exercise. \qed
	\end{proof}
	Since the model checking problem of \HQPTLP is undecidable, we focus on HyperQPTL to express $\omega$-regular hyperproperties. In particular, we show that HyperQPTL can express a range of relevant properties that are neither expressible in HyperLTL, nor in QPTL.

\subsection{The Expressiveness of HyperQPTL}
\label{sec:Expressiveness_HyperQPTL}
	HyperQPTL combines trace quantification with $\omega$-regularity. The interplay between the two features enables HyperQPTL to express a variety of properties. In Section~\ref{sec:intro}, we showed how HyperQPTL can express a form of promptness. In this section, we further elaborate on the type of properties HyperQPTL can express. In particular, we compare it to Prompt-LTL, a logic that extends LTL with bounded eventualities. Furthermore, HyperQPTL is also able to express epistemic properties by emulating the knowledge operator known from \LTLK.
	
	A straightforward class of properties HyperQPTL can express are $\omega$-regular properties over n-tuples of quantified traces. Formulas expressing this type of properties first have a trace quantifier prefix followed by a QPTL formula, i.e., they lie in the $Q_\pi^* Q_q^*$ fragment. This fragment of HyperQPTL corresponds to the extension of QPTL with \emph{prenex} trace quantification. However, the true expressive power of HyperQPTL originates from the fact that we allow the trace quantifiers and propositional quantifiers to alternate.
	
	\paragraph{Promptness Properties.}
	Promptness properties are an example for HyperQPTL's interplay between trace quantification and propositional quantification. Promptness expresses that eventualities are fulfilled within a bounded number of steps. One way to express promptness properties is the logic Prompt-LTL, which extends LTL with the promptness operator $\LTLdiamond_p$. A system satisfies a Prompt-LTL formula $\varphi$ if there is a bound $k$ such that all traces of the system fulfill the formula where each $\LTLdiamond_p$ in $\varphi$ is replaced by $\LTLdiamond^{\leq k}$, i.e., the system must fulfill all prompt eventualities within $k$ steps. For example, $\varphi = \G \F_p \psi $ holds in a system if there is a bound $k$ such that all traces of the system at all times satisfy $\psi$ within $k$ steps.	
	HyperQPTL can express a different type of promptness properties. In Section~\ref{sec:intro}, Formula~\ref{form:promptness}, we showed how one can state in HyperQPTL that there is a bound, common for all traces, until which an eventuality has to be fulfilled. The idea is to quantify a new proposition $b$, such that the first position in which $b$ is true serves as the bound. Compared to Prompt-LTL, HyperQPTL thus expresses a weaker form of promptness, while still being stronger than pure eventuality. This type of promptness only becomes meaningful when comparing several traces of the system: HyperQPTL can enforce that there is a common bound for all traces (the system cannot starve), but it does not make the bound explicit.
	The following example shows a more involved promptness property expressible in HyperQPTL.
	
	\begin{example}
		HyperQPTL can express \emph{bounded waiting for a grant}. It states that if the system requests access to a shared resource at point in time $t$, then it will be granted access within a bounded amount of time. The bound may depend on the point in time $t$ where access to the resource was requested. However, it may not depend on the current trace.	
		We express this property in HyperQPTL as follows, also adding that the system will not request access twice without being granted access in between.	
		\begin{align}
		\label{form:bounded_waiting1}
		&\forall \pi \ldot \LTLglobally (r_\pi \rightarrow \LTLnext (\neg r_\pi \LTLweakuntil g_\pi)) \\
		\label{form:bounded_waiting2}
		&\forall \pi \ldot \exists b \ldot \forall \pi' \ldot \LTLglobally (r_\pi \land r_{\pi'} \rightarrow \LTLnext(\LTLeventually b \land (\neg b ~ \LTLuntil g_\pi) \land (\neg b ~ \LTLuntil g_{\pi'})))
		\end{align}
		Formula \ref{form:bounded_waiting1} states that no second request is posed before being given a grant. Formula \ref{form:bounded_waiting2} expresses the bounded waiting property by universally quantifying a trace, then existentially quantifying a sequence of bounds $b$. Now, for every trace $\pi'$, whenever $\pi$ and $\pi'$ pose a request at the same point in time, both have to get access to the resource before $b$ holds next. Therefore, for each point in time, there is a bound such that all traces posing a request at that point in time get access within a bounded number of steps. Note that this property differs from saying ``all traces are eventually granted access'', where the bound may also depend on the trace under consideration. In this scenario, each of the infinitely many traces could wait arbitrarily long for the grant. In particular, it could happen that with each trace the waiting time is longer than before.
	\end{example}

	The above example shows how the interplay of trace quantifiers and propositional quantifiers can be leveraged to express a new class of promptness properties. We finally note that compared to Prompt-LTL, HyperQPTL cannot express that all eventualities must be fulfilled within a fixed $k$ number of steps.
	 
	 \begin{corollary}
	 	\label{cor:hyperqptl_incomparable_prompt}
	 	The expressiveness of HyperQPTL and Prompt-LTL is incomparable.
	 \end{corollary}
 
 	
	\paragraph{Epistemic Properties.}
	Another interesting class of properties that are not expressible in HyperLTL are epistemic properties. Epistemic properties describe the knowledge of agents that interact with each other in a system. Logics that express epistemic properties are often equipped with a so-called knowledge operator, e.g., \LTLK, which is LTL extended with the knowledge operator $\mathcal{K}_A \hspace{2pt} \varphi$. 	
	The operator denotes that an agent $A \subseteq \ap$ knows $\varphi$. An agent $A$ is characterized in terms of the atomic propositions he can observe. The semantics of the operator is described with the following rule
	\begin{align*}
		t, i \models \mathcal{K}_A \hspace{2pt} \varphi \quad \text{iff} \quad \forall t' \ldot t[0,i] =_A t'[0,i] \rightarrow t', i \models \varphi \enspace .
	\end{align*}	
	The formula is evaluated with respect to a trace $t$ and a position $i$. We omit the semantic definition for the rest of the logic, which corresponds to plain LTL. The semantic definition of the operator captures the idea that an agent knows some fact $\varphi$ if $\varphi$ holds on all traces that are indistinguishable for the agent.
	
	\begin{example}[Dining Cryptographers]
		\begin{figure}[t]
			\centering
			\begin{tikzpicture}[->,>=stealth',shorten >=1pt,auto,semithick,scale=1,transform shape]
			\node[state] (paid) {};
			\node[state,above left=1.4 and 1.4 of paid] (ca) {$C_1$};
			\node[state,above right=1.4 and 1.4 of paid] (cc) {$C_3$};
			\node[state,above=of paid] (cb) {$C_2$};
			\node[state,above=1.4 of cb] (env) {$\penv$};
			
			\draw (ca) edge[bend right] node[swap] {$\mathit{out}_1$} (paid)
			(cb) edge node {$\mathit{out}_2$} (paid)
			(cc) edge[bend left] node {$\mathit{out}_3$} (paid)
			(paid) edge node {$\mathit{paid}_\mathit{group}$} +(0,-1.2)
			(env) edge[bend right] node[swap,align=center,very near end] {$\mathit{paid}_1$,\\$s_{12}$, $s_{13}$} (ca)
			(env) edge[bend left] node[align=center,very near end] {$\mathit{paid}_3$,\\$s_{23}$, $s_{13}$} (cc)
			(env) edge node[align=center] {$\mathit{paid}_2$,\\$s_{12}$, $s_{23}$} (cb)
			(env) edge node {$\mathit{paid}_\mathit{NSA}$} +(2.5,0) 
			;
			\end{tikzpicture}
			\caption{The dining cryptographers problem with three cryptographers.}
			\label{fig:dining_cryptos}
		\end{figure}
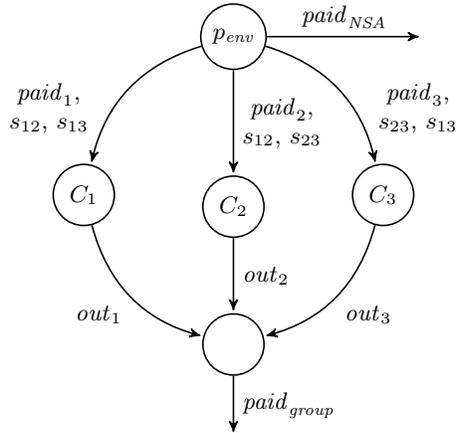
		The dining cryptographers problem~\cite{journals/cacm/Chaum85} is an interesting example of how epistemic properties can characterize non-trivial protocols. The problem describes the following situation (see Fig.~\ref{fig:dining_cryptos}): three cryptographers $C_1,C_2,$ and $C_3$ sit at a table in a restaurant and either one of cryptographers or, alternatively, the NSA paid for their meal. The task for the cryptographers is to figure out whether the NSA or one of the cryptographers paid. However, if one of the cryptographers paid, then the others must not be able to infer who it was. Each cryptographer $C_i$ receives several bits of information: $\mathit{paid_i}$ indicating whether or not he pays the bill, and two secrets, each shared with one of the other cryptographers. The secrets can be used to encode the information they share as output $\mathit{out}_i$. By combining the outputs of all cryptographers, it must become clear whether the NSA or one of the group paid.
		The specification of the protocol can be easily formalized in \LTLK. The following formula  describes the desired behavior of agent $C_1$:
		\begin{align*}
		\mathit{DC}&\mathit{agent1} \coloneqq \\
		&(\mathit{paid}_\mathit{group} \land \neg \mathit{paid}_1
		\rightarrow (\K{C_1} (\mathit{paid}_2 \lor \mathit{paid}_3)  \land
		\neg \K{C_1} \mathit{paid}_2 \land
		\neg \K{C_1} \mathit{paid}_3
		) )\\
		\land ~ &(\mathit{paid}_\mathit{NSA} \rightarrow \K{C_1} (\neg\mathit{paid}_1 \land \neg\mathit{paid}_2 \land \neg\mathit{paid}_3)) \enspace .
		\end{align*}
	\end{example}
	The knowledge operator can also be defined for hyperlogics~\cite{MarkusThesis}. It receives an additional parameter $\pi$, indicating the trace the knowledge refers to. When added to HyperQPTL, it has the following semantics:
	\begin{alignat*}{3}
	&\pathassign, i \models_T \mathcal{K}_{A,\pi} \varphi \quad &&\text{iff} && \quad \forall t' \in T \ldot \pathassign(\pi)[0, i] =_A t'[0, i] \rightarrow \pathassign[\pi \mapsto t'], i \models_T \varphi \enspace .
	\end{alignat*}
	The knowledge operator, however, can be encoded in HyperQPTL using propositional quantification. Epistemic problems, such as the dining cryptographers problem, can thus be expressed in HyperQPTL.
	\begin{theorem}[\cite{MarkusThesis}~]  \label{thm:subsumes-epistemic}
		HyperQPTL can emulate the knowledge operator.
	\end{theorem}
	\begin{proof}
		We recap the proof from~\cite{MarkusThesis}: Let $\varphi = Q_{\pi\slash q} \ldots Q_{\pi\slash q} \ldot \varphi'$ be a HyperQPTL formula, equipped with the knowledge operator as defined above. We assume that $\varphi$ is given in negated normal form, i.e. each $\mathcal{K}_{A,\pi}$ occurs either in positive position or in negated form. Let $u$ and $t$ be fresh propositions and let $\pi'$ be a fresh trace variable. Recursively, we replace each knowledge operator $\mathcal{K}_{A,\pi}$ occurring in $\varphi$ in positive position with the following formula
		\begin{align*}
		& Q_{\pi\slash q} \ldots Q_{\pi\slash q} \ldot  \exists u \ldot \forall r \ldot \forall \pi' \ldot ~  \varphi'[{\mathcal{K}_{A,\pi} \psi} \mapsto u] ~\wedge\\
		& \quad ((r~ \LTLuntil ~(u\wedge r \wedge \LTLnext \LTLsquare \neg r)) \wedge \LTLsquare (r \rightarrow A_\pi = A_{\pi'}) \rightarrow \LTLsquare (r \wedge \LTLnext \neg r \rightarrow \psi[\pi \mapsto \pi']))
		\end{align*}
		and each $\mathcal{K}_{A,\pi}$ occurring negatively with the following formula
		\begin{align*}
		& Q_{\pi\slash q} \ldots Q_{\pi\slash q} \ldot  \exists u \ldot \forall r \ldot \exists \pi' \ldot ~  \varphi'[{\neg\mathcal{K}_{A,\pi}} \psi \mapsto u] ~ \wedge \\
		& \quad ((r~ \LTLuntil ~(u\wedge r \wedge \LTLnext \LTLsquare \neg r)) \rightarrow \LTLsquare (r \rightarrow A_\pi = A_{\pi'}) \wedge\LTLsquare (r \wedge \LTLnext \neg r \rightarrow \neg \psi[\pi \mapsto \pi'])),
		\end{align*}
		where we use $\varphi'[{\mathcal{K}_{A,\pi}} \psi \mapsto u]$ to denote that in $\varphi'$, \emph{a single} occurrence of the knowledge operator is replaced by $u$, and $\psi[\pi \mapsto \pi']$ to denote the formula where $\pi$ is replaced by $\pi'$. The existentially quantified proposition $u$ indicates the points in time where the knowledge operator is supposed to hold/not hold. The universally quantified proposition $r$ is assumed to change once from $r$ to $\neg r$ and thereby point at one of the points in time picked by $u$. It is then used to compare the prefix of the old trace $\pi$ and an alternative trace quantified by the trace variable $\pi'$. \qed
	\end{proof}

\section{HyperQPTL Realizability}
\label{sec:HyperQPTL_realizability}
	In reactive synthesis, the task is, given a specification $\varphi$, to construct a system that satisfies the specification. More precisely, the system is assumed to receive some inputs from an environment and has to react with outputs such that the specification is fulfilled. The realizability problem asks for the existence of a so-called \emph{strategy tree}, where the edges are labeled with all possible inputs and the task is to find a function $f$ that labels the nodes with the corresponding outputs. Figure~\ref{fig:strategy_tree} shows a strategy tree for a single input bit $i$.
	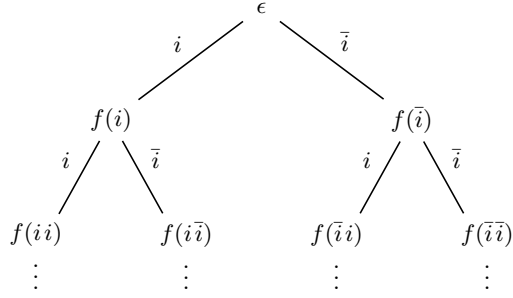
\begin{figure}
		\centering
		\tikzstyle{treenode} = [text width = 3em, inner sep=3pt, text centered]
		\begin{tikzpicture}[
			semithick,
			level/.style={level distance=10cm},
			level 1/.style={sibling distance=4cm, level distance=1.5cm},
			level 2/.style={sibling distance=2cm, level distance=1.8cm}
		]
		
			\node[treenode] {$\epsilon$}
			child {
				node[treenode] {$f(i)$}        
				child {
					node[treenode] {$f(i\hspace{1pt}i)$ \\ $\vdots$}
					edge from parent
					node[above left] {$i$}
				}
				child {
					node[treenode] {$f(i \hspace{1pt} \bar{i})$ \\ $\vdots$}
					edge from parent
					node[above right] {$\bar{i}$}
				}
				edge from parent 
				node[above left] {$i$}
			}
			child {
				node[treenode] {$f(\bar{i})$}        
				child {
					node[treenode] {$f(\bar{i} \hspace{1pt} i)$ \\ $\vdots$}
					edge from parent
					node[above left] {$i$}
				}
				child {
					node[treenode] {$f(\bar{i} \hspace{1pt} \bar{i})$ \\ $\vdots$}
					edge from parent
					node[above right] {$\bar{i}$}
				}
				edge from parent         
				node[above right] {$\bar{i}$}
			};
		\end{tikzpicture}
		\caption{A strategy tree for the reactive realizability problem.}
		\label{fig:strategy_tree}
	\end{figure}	
	We define strategies following~\cite{conf/cav/FinkbeinerHLST18}.
		Let a set $\ap = I \cupdot O$ be given.
		A \emph{strategy} $f \colon \strat{I}{O}$  maps sequences of input valuations $\pow{I}$ to an output valuation $\pow{O}$.
		For an infinite word $w = w_0 w_1 w_2 \cdots \in (\pow{I})^\omega$, the trace corresponding to a strategy $f$ is defined as $(f(\epsilon) \cup w_0)(f(w_0) \cup w_1)(f(w_0 w_1) \cup w_2)\ldots \in (2^{I \cup O})^\omega$.
		For any trace $w = w_0 w_1 w_2 \ldots \in (2^{I \cup O})^\omega$ and strategy $f \colon \strat{I}{O}$, we lift the set containment operator $\in$ defining that $w \in f$ iff $f(\epsilon) = w_0 \cap O$ and $f((w_0 \cap I) \cdots (w_i \cap I)) = w_{i+1} \cap O$ for all $i \geq 0$. We say that a strategy $f$ satisfies a HyperQPTL formula $\varphi$ over $\ap = I \cupdot O$ iff $\set{w ~|~ w \in f}$ satisfies $\varphi$.
	
	With the definition of a strategy at hand, we can define the realizability problem of HyperQPTL formally.
	\begin{definition}[HyperQPTL Realizability]
		A HyperQPTL formula $\varphi$ over atomic propositions $\ap = I \cupdot O$ is realizable if there is a strategy $f \colon \strat{I}{O}$ that satisfies $\varphi$.
	\end{definition}
	For technical reasons, we assume (without loss of generality) that quantified atomic propositions are classified as outputs, not inputs. This complies with the intuition that propositional quantifiers should be a means for additional expressiveness; they should not overwrite the inputs received from the environment.
	The definition of realizability of QPTL and HyperLTL specifications is inherited from the definition for HyperQPTL.
	
	Compared to the standard realizability problem, the distributed realizability problem is defined over an architecture, containing a number of processes interacting with each other. The goal is to find a strategy for each of the processes. In the following proofs, we will make use of the distributed realizability problem of QPTL, which we therefore also define formally.
	
		A \emph{distributed architecture}~\cite{conf/focs/PnueliR90, conf/lics/FinkbeinerS05} $A$ over atomic propositions $\ap$ is a tuple $\langle P,p_\mathit{env},\mathcal{I},\mathcal{O} \rangle$, where $P$ is a finite set of processes and $p_\mathit{env} \in P$ is a designated environment process. The functions $\mathcal{I}: P \rightarrow 2^\ap$ and $\mathcal{O}: P \rightarrow 2^\ap$ define the inputs and outputs of processes. The output of one process can be the input of another process. The output of the processes must be pairwise disjoint, i.e., for all $p \ne p' \in P$ it holds that $\mathcal{O}(p) \cap \mathcal{O}(p') = \emptyset$. We assume that the environment process forwards inputs to the processes and has no input of its own, i.e., $\mathcal{I}(p_\mathit{env}) = \emptyset$.
	
	\begin{definition}[Distributed QPTL Realizability~\cite{conf/lics/FinkbeinerS05}]
		A QPTL formula $\varphi$ over free atomic propositions $\ap$ is realizable in an architecture $A = \langle P,p_\mathit{env},\mathcal{I},\mathcal{O} \rangle $ if for each process $p \in P$, there is a strategy $f_p \colon \strat{\mathcal{I}(p)}{\mathcal{O}(p)}$ such that the combination of all $f_p$ satisfies $\varphi$.
	\end{definition}
The distributed realizability problem for QPTL is (inherited from LTL) in general undecidable~\cite{conf/focs/PnueliR90}. However, we will use the result that the problem remains decidable for architectures without \emph{information forks}\cite{conf/lics/FinkbeinerS05}. The notion of  information forks captures the flow of data in the system.
Intuitively, an architecture contains an information fork if the processes cannot be ordered linearly according to their informedness.
Formally, an information fork in an architecture $A = \langle P,p_\mathit{env},\mathcal{I},\mathcal{O} \rangle $ is defined as a tuple $(P',V', p, p')$, where $p,p'$ are two different processes, $P' \subseteq P$, and $V' \subseteq AP$ is disjoint from $\mathcal{I}(p) \cup \mathcal{I}(p')$.
$(P',V', p, p')$ is an information fork if $P'$ together with the edges that are labeled with at least one variable from $V'$ forms a subgraph rooted in the environment and there exist two nodes $q,q' \in P'$ that have edges to $p,p'$, respectively, such that $\mathcal{O}(q)\cap\mathcal{I}(p) \nsubseteq \mathcal{I}(p')$ and $\mathcal{O}(q')\cap\mathcal{I}(p') \nsubseteq \mathcal{I}(p)$. The definition formalizes the intuition that $p$ and $p'$ receive incomparable input bits, i.e., they have incomparable information.
	\begin{example}\label{ex:info_fork}
		Two example architectures are depicted in Fig.~\ref{fig:distributed-architectures}\cite{conf/cav/FinkbeinerHLST18}. The processes in Fig.~\ref{fig:informationFork1} receive distinct inputs and thus neither process is more informed than the other. The architecture therefore contains an information fork with $P' = \set{\mathit{env}, p, p'}, V' = \set{i, i'}, q = \mathit{env}, q' = \mathit{env}$.
		The processes in Fig.~\ref{fig:architecture-incomplete-information} can be ordered linearly according to the subset relation on the inputs and thus the architecture contains no information fork.
		
		\begin{figure}[t]
			\begin{subfigure}[t]{0.49\columnwidth}
				\centering
				\begin{tikzpicture}[->,>=stealth',shorten >=1pt,auto,semithick,scale=1,transform shape,scale=0.8]
				\tikzstyle{every state}=[shape=rectangle]
				\tikzstyle{envstate}=[shape=circle,scale=.95]
				
				\node [state,envstate] (e) {$env$};
				\node [state, below left=1 of e] (a) {$p$};
				\node [state, below right=1 of e] (b) {$p'$};
				\path[->]
				(e) edge node [label,above left = 0 and -0.1] {$i$} (a)
				(e) edge node [label,above right = 0 and -0.1] {$i'$} (b)
				(a) edge node [label,above left = -0.15 and 0] {$o$} +(0,-1.3)
				(b) edge node [label,above right = -0.15 and 0] {$o'$} +(0,-1.3)
				;
				\end{tikzpicture}
				\caption{Information fork: An architecture with two processes; process $p$ to produces output $o$ from input $i$ and $p'$ produces output $o'$ from input $i'$.}
				\label{fig:informationFork1}
			\end{subfigure}\hfill
			\begin{subfigure}[t]{0.49\textwidth}
				\centering
				\begin{tikzpicture}[->,>=stealth',shorten >=1pt,auto,semithick,scale=1,transform shape,scale=0.8]
				\tikzstyle{every state}=[shape=rectangle]
				\tikzstyle{envstate}=[shape=circle,scale=.95]
				
				\node [state,envstate] (e) {$env$};
				\node [state, below left=1 of e] (a) {$p$};
				\node [state, below right=1 of e] (b) {$p'$};
				\path[->]
				(e) edge node [label,above left = 0 and -0.1] {$i$} (a)
				(e) edge node [label,above right = 0 and -0.1] {$i,i'$} (b)
				(a) edge node [label,above left = -0.15 and 0] {$o$} +(0,-1.3)
				(b) edge node [label,above right = -0.15 and 0] {$o'$} +(0,-1.3)
				;
				\end{tikzpicture}
				\caption{No information fork: The same architecture as on the left, where the inputs of process $p'$ are changed to $i$ and $i'$.}
				\label{fig:architecture-incomplete-information}
			\end{subfigure}
			\caption{Distributed architectures}
			\label{fig:distributed-architectures}
		\end{figure}
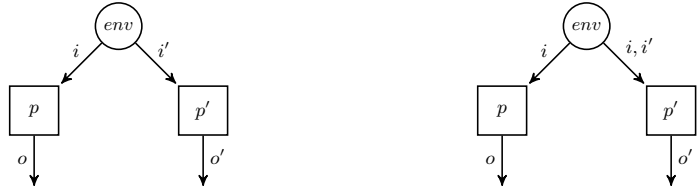
	\end{example}
	In the following sections, we identify tight syntactic fragments of HyperQPTL for which the standard realizability problem is decidable. We give decidability proofs and show that formulas outside the decidable fragments are in general undecidable. An important aspect for decidability is the number of universal trace quantifiers that appear in the formula. We thus present our findings in three categories, depending on the number of universal trace quantifiers a formula has.
	
	\subsection{No Universal Trace Quantifier}
	\label{sec:NoUnivTrace}
	We show that the realizability problem of any HyperQPTL formula without a $\forall_\pi$ quantifier is decidable. The problem is reduced to QPTL realizability.
	\begin{theorem}
		\label{thm:dec_no_univ_trace_quantifier}
		Realizability of the $(\exists_{\pi}^*Q_q^*)^*$ fragment of HyperQPTL is decidable.
	\end{theorem}
	\begin{proof}
		Let a $(\exists_{\pi}^*Q_{q}^*)^*$ HyperQPTL formula $\varphi$ over $\ap = I \cupdot O = \{a^0,\ldots,a^{k}\}$ with trace quantifiers $\pi_0, \ldots \pi_n$ be given. 
		We reduce the problem to the realizability problem of QPTL, which is known to be decidable (since QPTL formulas can be translated to Büchi automata). The idea is to replace each existential trace quantifier $\exists \pi_i$ with quantification of propositions $a^0_{\pi_i}, a^1_{\pi_i}, \ldots , a^k_{\pi_i}$, one for each $a^j \in \ap$, thereby mimicking the quantification of a trace. To make sure that only traces from an actual strategy tree are chosen, we add a dependency formula which forces the outputs to be dependent on the inputs. The following QPTL formula implements the idea.
		\begin{align*}
			\varphi_\mathit{QPTL}  \coloneqq \enspace & \varphi [i \leq n: \exists \pi_i \mapsto \exists a^0_{\pi_i} \ldot \ldots \exists a^k_{\pi_i} \ldot] \enspace \land \\ 
			& \quad \bigwedge_{i\leq n} \enspace \bigwedge_{j \leq n} (I_{\pi_i} \neq I_{\pi_j}) \R (O_{\pi_i} = O_{\pi_j})
		\end{align*}
		We use the notation $[i \leq n: \exists \pi_i \mapsto \exists a^0_{\pi_i} \ldot \ldots \exists a^k_{\pi_i} \ldot]$ to indicate that each $\pi_i$ for $0 \leq i \leq n$ is replaced with the respective series of existential propositional quantification. Furthermore, we write $I_{\pi_i} \neq I_{\pi_j}$ as syntactic sugar for $\bigvee_{a \in I} a_{\pi_i} \nleftrightarrow a_{\pi_j}$ (and similarly for $O_{\pi_i} = O_{\pi_j}$).
		We show that $\varphi$ and $\varphi_\mathit{QPTL}$ are equirealizable.
		For the first direction, assume that $\varphi$ is realizable by a strategy $f$. Notice that all atomic propositions in $\varphi_\mathit{QPTL} $ are bound by a propositional quantifier. Therefore, if the witness sequences for the quantified propositions can be chosen correctly, any strategy realizes $\varphi_\mathit{QPTL}$. Propositions $a_{\pi_i}^j$ are chosen according to the witness traces of $f \models \varphi$. Witnesses for the remaining atomic propositions are also chosen according to their witnesses from $f \models \varphi$. Now, the first conjunct of $\varphi_\mathit{QPTL}$ is fulfilled since $f \models \varphi$ holds. The second conjunct is fulfilled since any two traces $\pi_i, \pi_j$ of a strategy tree fulfill by construction $(I_{\pi_i} \neq I_{\pi_j}) \R (O_{\pi_i} = O_{\pi_j})$.
		For the other direction, assume that $\varphi_\mathit{QPTL}$ is realizable (by construction independently from the strategy). Let $t_{a^0_{\pi_0}}, \ldots, t_{a^k_{\pi_n}}$ be the witness sequences for the respective quantified atomic propositions.		
		The following strategy realizes $\varphi$.
		$$
			f(\sigma)= 
			\begin{cases}
				\set{t_{a_{\pi_i}} [| \sigma |] \mid a \in O} \quad & \text{if for some } i \leq n,\\
					& \quad \sigma = \set{t_{a_{\pi_i}} [0] \mid a \in I} \ldots  \set{t_{a_{\pi_i}} [| \sigma |] \mid a \in I} \\
				\emptyset              			 & \text{otherwise}
			\end{cases}
		$$
		Strategy $f$ chooses the outputs according to the witnesses for the propositions encoding the traces. Note that because of the second conjunct in $\varphi_\mathit{QPTL}$, the output is always unique, even if several encoded traces start with the same input sequence. Now, $f \models \varphi$ holds because of the first conjunct of $\varphi_\mathit{QPTL}$. \qed
	\end{proof}

	
	\subsection{Single Universal Trace Quantifier}
	\label{sec:SingleUnivTrace}
	In this fragment, we allow exactly one universal trace quantifier. It is particularly interesting as it contains many promptness properties. For example, the following promptness formulation mentioned in the introduction lies within the fragment:
	$$
	\exists b. \forall \pi.~ \F b \land (\neg b ~ \U e_\pi) \enspace .
	$$
	
	\begin{theorem}
		\label{thm:dec_single_trace_quantifier}
		Realizability of the $\exists_{q/\pi}^*\forall_q^*\forall_{\pi}Q_q^*$ fragment is decidable.
	\end{theorem}
	We show the theorem in two steps. First, we generalize a proof from \cite{conf/cav/FinkbeinerHLST18}, showing that realizability of the $\exists_{\pi}^* \forall_{\pi} Q_q^*$ fragment is decidable. Second, we show that we can reduce the realizability problem of any HyperQPTL formula to a formula where some propositional quantifiers are replaced with trace quantifiers.
	
	\begin{lemma}
		\label{lem1:dec_single_trace_quantifier}
		Realizability of the $\exists_{\pi}^* \forall_{\pi} Q_q^*$ fragment is decidable.
	\end{lemma}
	\begin{proof}
		The reasoning generalizes the proof in~\cite{conf/cav/FinkbeinerHLST18} showing that realizability $\exists_\pi^* \forall_\pi$ HyperLTL formulas is decidable. We reduce the problem to the distributed realizability problem of QPTL without information forks, which is --- since QPTL is subsumed by the $\mu$-calculus --- decidable~\cite{conf/lics/FinkbeinerS05}. Let a HyperQPTL formula $\varphi = \exists\pi_1 \ldot \ldots  \exists \pi_n \ldot \forall \pi \ldot \psi$ over $AP = I \cupdot O$ be given, where $\psi$ is from the $Q^*_q$ fragment.
		We define a distributed architecture $\mathcal{A}$ over an extended set of atomic propositions $\ap' = I \cup O \cup I' \cup O'$. Similarly to the proof in Theorem~\ref{thm:dec_no_univ_trace_quantifier}, $I'$ and $O'$ are composed of a copy of the atomic propositions for each existentially quantified variable $\pi_j$. Formally, $I' = \bigcup_{1 \leq j \leq n} \set{i_{\pi_j} \mid i \in I}$ and $O' = \bigcup_{1 \leq j \leq n} \set{o_{\pi_j} \mid o \in O}$.
		Now we define $\mathcal{A}$ as follows.
		\begin{align*}
			\mathcal{A} &\coloneqq \langle (p_\mathit{env}, p_1, p_2), p_\mathit{env}, \mathcal{I}, \mathcal{O}, \rangle \\
			\mathcal{I} &\coloneqq  (p_1 \mapsto \emptyset, p_2 \mapsto I) \\
			\mathcal{O} &\coloneqq (p_\mathit{env} \mapsto I, p_1 \mapsto I' \cup O', p_2 \mapsto O)
		\end{align*}
		\begin{figure}[t]
			\centering
			\begin{tikzpicture}[->,>=stealth',shorten >=1pt,auto,semithick,scale=1,transform shape,scale=0.8]
				\tikzstyle{every state}=[shape=rectangle]
				\tikzstyle{envstate}=[shape=circle,scale=.95]
				
				\node [state,envstate] (e) {$env$};
				\node [state, below =0.8 of e] (a) {$p_2$};
				\node [state, left=1 of a] (b) {$p_1$};
				\path[->]
				(e) edge node [label,above right] {$I$} (a)
				(a) edge node [label, right] {$O$} +(0,-1.3)
				(b) edge node [label, left] {$I' \cup O'$} +(0,-1.3)
				;
			\end{tikzpicture}
			\caption{Distributed architecture encoding existential choice of traces.}
			\label{fig:architecture-proof}
		\end{figure}
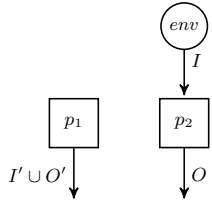
		The architecture is displayed in Fig.~\ref{fig:architecture-proof}.
		The idea is that process $p_1$ sets the values of all $i_{\pi_j}$ and $o_{\pi_j}$ (for $j \leq n$) and thereby determines the choice for the existentially quantified traces. Process $p_1$ receives no input and therefore needs to make a deterministic choice. Process $p_2$ then solves the realizability of formula $\forall \pi \ldot \psi$. The following QPTL formula $\varphi'$ encodes the idea.
		\begin{align*}
			\varphi' \coloneqq \psi' \land \lparen \bigwedge_{1 \leq j \leq n} ( I_{\pi_j} \neq I) \LTLrelease (O_{\pi_j} = O) \rparen \enspace ,
		\end{align*}
		where $\psi'$ is defined as $\psi$, where all $a_\pi$ are replaced by $a$ (but atomic propositions $a_{\pi_j}$ are still part of $\psi'$!). Note that QPTL formulas implicitly quantify over all traces universally.
		Similarly to the proof in Theorem~\ref{thm:dec_no_univ_trace_quantifier}, the second conjunct ensures that process $p_1$ encodes actual paths from the strategy tree of process $p_2$ (which is also the strategy tree for formula $\varphi$). Thus, $\varphi'$ is realizable for the distributed architecture $\mathcal{A}$ iff $\varphi$ is realizable. \qed
	\end{proof}
	To state the second lemma, we need to define what it means to replace quantifiers in a formula.
	%
		Let $\varphi = Q_{\pi/q}, \ldots , Q_{\pi/q} \ldot \psi$ be a HyperQPTL formula, and $J$ be a set of indices such that for all $j \in J$, there exists a propositional quantifier $\exists q_j$ or $\forall q_j$ in $\varphi$. Furthermore, assume that no $\pi_j$ with $j \in J$ occurs in $\varphi$ and that $a \in \ap$. We denote by $\varphi [J \hookrightarrow_a \pi]$ the formula where each propositional quantifier $\exists q_j$ (or $\forall q_j$, respectively) with $j \in J$ is replaced with the corresponding trace quantifier $\exists \pi_{j}$ (or $\forall \pi_{j}$, respectively); and each $q_j$ in $\psi$ is replaced by $\mathit{a}_{\pi_{j}}$.

	\begin{lemma}
		\label{lem2:dec_single_trace_quantifier}
		Let any HyperQPTL formula $\varphi$ over $\ap = I \cupdot O$ and a set of indices $J$ be given. If $\varphi [J \hookrightarrow_i \pi]$ is realizable, then so is $\varphi$, where $i \in I$ is an arbitrary input, assuming w.l.o.g., that $I$ is non-empty.
	\end{lemma}
	\begin{proof}
		Let $\varphi$ and $J$ be given. Formula $\varphi [J \hookrightarrow_i \pi]$ replaces the quantification over sequences $(2^{\{q\}})^\omega$ with trace quantification, where the trace is only used for statements about a single input $i$. We thus exploit the fact that in the realizability problem, there is a trace for every input sequence. Therefore, the transformed formula is equirealizable. \qed
	\end{proof}
	Now, we have everything we need to prove Theorem~\ref{thm:dec_single_trace_quantifier}.
	\begin{proof}[of Theorem~\ref{thm:dec_single_trace_quantifier}]
		Let $\varphi$ be a HyperQPTL formula of the $\exists_{q/\pi}^*\forall_q^*\forall_{\pi}Q_q^*$ fragment. First, observe that in the quantifier prefix of $\varphi$, the $\forall_q^*$ quantifiers and the $\forall_{\pi}$ can be swapped. The resulting formula belongs to the $\exists_{q/\pi}^*\forall_{\pi}Q_q^*$ fragment. By Lemma~\ref{lem2:dec_single_trace_quantifier}, the formula can be transformed to a equirealizable formula of the $\exists_{\pi}^*\forall_{\pi}Q_q^*$ fragment, for which realizability is decidable by Lemma~\ref{lem1:dec_single_trace_quantifier}. \qed
	\end{proof}
	Lemma~\ref{lem2:dec_single_trace_quantifier} allows us to decide realizability of a HyperQPTL formula by replacing propositional quantifiers with trace quantifiers. Thus, we can reduce HyperQPTL realizability to HyperLTL realizability, a fact that we use in Section~\ref{sec:Experiments} to describe a bounded synthesis algorithm for HyperQPTL.
	\begin{corollary} \label{cor:hyperqptl_2_hyperltl}
		The realizability problem of HyperQPTL can be soundly reduced to the realizability problem of HyperLTL.
	\end{corollary}

	Lastly, we show that the decidable fragment is tight in the class of formulas with a single universal trace quantifier. We do so by showing that a propositional $\forall^*_q \exists^*_q$ quantifier alternation followed by a single trace quantifier $\forall_{\pi}$ leads to an undecidable realizability problem. The proof is carried out by a reduction from Post's Correspondence Problem.
	\begin{theorem}
		Realizability is undecidable for HyperQPTL formulas with a single $\forall_\pi$ quantifier outside the $\exists_{q/\pi}^*\forall_q^*\forall_{\pi}Q_q^*$ fragment.
	\end{theorem}
	\begin{proof}
		Inherited from HyperLTL, realizability of formulas with a $\forall_\pi$ quantifier followed by an $\exists_\pi$ quantifier is undecidable~\cite{conf/cav/FinkbeinerHLST18}. It remains to show that realizability of formulas from the $\forall_q^*\exists_q^*\forall_\pi$ fragment is in general undecidable. We give a reduction from  Post's Correspondence Problem (PCP)~\cite{post} to a HyperQPTL formula from the $\forall_q^* \exists_q^* \forall_\pi$ fragment.
		In PCP, we are given two equally long lists $\alpha$ and $\beta$ consisting of finite words from some alphabet $\Sigma$ of size $n$.
		PCP is the problem to find an index sequence $(i_k)_{1 \leq k \leq K}$ with $K \geq 1$ and $1 \leq i_k \leq n$, such that $\alpha_{i_1}\dots\alpha_{i_K}=\beta_{i_1}\dots\beta_{i_K}$.
		Intuitively, PCP is the problem of choosing an infinite sequence of domino stones (with finitely many different stones), where each stone consists of two words $\alpha_i$ and $\beta_i$.
		Let a PCP instance with $\Sigma = \{a_1,a_2,..., a_n\}$ and two lists $\alpha$ and $\beta$ be given.
		We choose our set of atomic propositions as follows: $\ap \coloneqq I \cupdot O$ with $I := \{i\}$  and $O \coloneqq (\Sigma \cup \{\dot a_1, \dot a_2,... ,\dot a_n\} \cup {\#})^2$, where we use the dot symbol to encode that a stone starts at this position of the trace. We write $\tilde{a}$ to denote either $a$ or $\dot{a}$. The single input $i$ spans a binary strategy tree.
		We encode the PCP instance into a HyperQPTL formula that is realizable if and only if the PCP instance has a solution:	
		\begin{align*}
		\forall q_i \ldot \forall \vec{q} \ldot \exists p_i \ldot \exists \vec{p} \ldot \forall \pi .~
		&((\LTLsquare \pi = p_i) \rightarrow (\LTLsquare \pi = \vec{p})) ~ \wedge \\
		& ((\LTLsquare \pi = (q_i, \vec{q}) ) \rightarrow \varphi_\mathit{reduc}(q_i, \vec{q}, p_i, \vec{p})) \enspace ,
		\end{align*}
		where $\vec{q}$ and $\vec{p}$ are sequences of universally and existentially quantified propositional variables, such that for each $(o,o') \in O$, there is a $q_{(o,o')} \in \vec{q}$ and a $p_{(o,o')} \in \vec{p}$. Together with $q_i$ and $p_i$ for the input $i$, they simulate a universally and an existentially quantified trace from the model.
		The notation $\pi = \vec{q}$ denotes that for every $q_a \in \vec{q}$, it holds that $a_\pi \leftrightarrow q_a$.
		As seen before, the premise $(\LTLsquare \pi = (q_i, \vec{q}) )$ and the conjunct $(\LTLsquare \pi = p_i) \rightarrow (\LTLsquare \pi = \vec{p})$ ensure that the propositions $(q_i, \vec{q})$ and $(p_i, \vec{p})$ are  chosen to represent actual traces from the model. The universal quantification $\pi$ thus only ensures that $(q_i, \vec{q})$ and $(p_i, \vec{p})$, which are used for the main reduction, are chosen correctly.
		The reduction is implemented in the formula $\varphi_\mathit{reduc}$ and follows the construction in~\cite{conf/concur/FinkbeinerH16}, where it is shown that the satisfiability and realizability problem of HyperLTL are undecidable for a $\forall \exists$ trace quantifier prefix.
		\begin{align*}
		\varphi_\mathit{reduc}(q_i, \vec{q}, p_i, \vec{p}) := \;&\varphi_\mathit{rel}(q_i) \rightarrow \varphi_\mathit{is++}(q_i,p_i)\\
		&{} \wedge \varphi_\mathit{start}(\varphi_\mathit{stone\&shift}(\vec{q},\vec{p}),q_i) \wedge \varphi_\mathit{sol}(q_i,\vec{q})
		\end{align*}
		\begin{itemize}
			\item  $\varphi_\mathit{rel}(q_i) := \neg q_i \U \G q_i$ defines the set of \emph{relevant} traces trough the binary strategy tree (see Fig.~\ref{fig:relevant_traces}).
			\begin{figure}[t]
				\centering
				\tikzstyle{treenode} = [text width = 3em, inner sep=0pt, text centered]
				\begin{tikzpicture}[
				semithick,
				level 1/.style={level distance=1.5cm, sibling distance=4.5cm},
				level 2/.style={level distance=1cm, sibling distance=3cm},
				level 3/.style={level distance=0.75cm, sibling distance=2.25cm}
				]
				\pgfdeclarelayer{bg}   
				\pgfsetlayers{bg,main}
				
				\node[treenode] (L0) {$\circ$}
				child {
					node[treenode] (L1) {$\circ$}        
					child {
						node[treenode] (L11) {$\circ$}
							child{
							node[treenode] (L111) {$\circ$}
							edge from parent
							node[above] {$i$}
							}
							child{
								node[treenode] (L112) {$\circ$}
								edge from parent
								node[above] {$\bar{i}$}
							}
					edge from parent
					node[above]{$i$}
					}
					child {
						node[treenode] (L12) {$\circ$}
						edge from parent
						node[above] {$\bar{i}$}
					}
					edge from parent 
					node[above] {$i$}
				}
				child {
					node[treenode] (L2) {$\circ$}        
					child {
						node[treenode] (L21) {$\circ$}
						child{
							node[treenode] (L211) {$\circ$}	
							edge from parent
							node[above]{$i$}
						}
						child{
							node[treenode] (L212) {$\circ$}	
							edge from parent
							node[above]{$\bar{i}$}
						}
						edge from parent
						node[above] {$i$}
					}
					child {
						node[treenode] (L22) {$\circ$}
						child{
							node[treenode] (L221) {$\circ$}
							edge from parent
							node[above] {$i$}
						}
						child{
							node[treenode] (L222) {$\circ$}
							edge from parent
							node[above] {$\bar{i}$}
						}
						edge from parent
						node[above] {$\bar{i}$}
					}
					edge from parent         
					node[above] {$\bar{i}$}
				};
			
				\begin{pgfonlayer}{bg}
					\draw[color=green!50, double = green!30, double distance=5mm, line cap=round] ([yshift=2pt, xshift=-2pt]L0.center) to ([yshift=2pt, xshift=-2pt]L111.center);
					
					\draw[color=green!50, double = green!30, double distance=5mm, line cap=round] ([yshift=2pt, xshift=-2pt]L2.center) to ([yshift=2pt, xshift=-2pt]L211.center);
					
					\draw[color=green!50, double = green!30, double distance=5mm, line cap=round] ([yshift=2pt, xshift=-2pt]L22.center) to ([yshift=2pt, xshift=-2pt]L221.center);
				\end{pgfonlayer}

				\end{tikzpicture}
				\caption{A sketch of the strategy tree of our PCP reduction: relevant traces are marked in green.}
				\label{fig:relevant_traces}
			\end{figure}
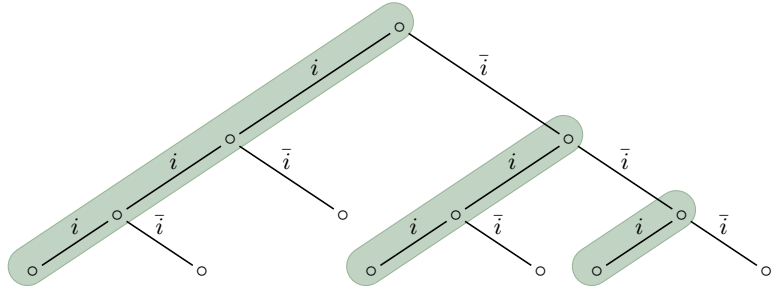
			
			\item $\varphi_\mathit{is++}(q_i,p_i) := (\neg q_i \wedge \neg p_i){} \U{} (\G q_i \wedge \neg p_i \wedge \X \G p_i)$ defines that a relevant trace is the direct successor trace of another relevant trace.
			
			\item $\varphi_{\text{sol}}(q_i,\vec{q}) \coloneqq \G q_i \rightarrow ((\bigvee_{i=1}^{n}q_{(\dot a_i, \dot a_i)}) \wedge (\bigvee_{i=1}^{n}q_{(\tilde a_i, \tilde a_i)}))$ $\U \G q_{(\#, \#)}$ ensures that the path on which globally $i$ holds is a ``solution'' trace, i.e., encodes the PCP solution sequence.
			
			\item $\varphi_\mathit{start}(\varphi, q_i):= \neg q_i \U (\varphi \wedge \G q_i)$ cuts off an irrelevant prefix until $\varphi$ starts.
			
			\item $\varphi_\mathit{stone\&shift}(\vec{q},\vec{p})$ encodes that the trace simulated by $\vec{q}$ starts with a valid encoding of a stone from the PCP instance and that the trace simulated by $\vec{p}$ encodes the same trace but with the first stone removed (see~\cite{conf/concur/FinkbeinerH16}). 
		\end{itemize}
		For example, let $\alpha$ with $\alpha_1 = a$, $\alpha_2 = ab$, $\alpha_3 = bba$, and $\beta$ with $\beta_1 = baa$, $\beta_2 = aa$ and $\beta_3 = bb$ be given. A possible solution for this PCP instance is be $(3,2,3,1)$, since $bbaabbbaa = i_\alpha = i_\beta$.
		The full sequence at the trace $\G i$ represents the solution with the outputs 
		$$(\dot b, \dot b)(b, b)(a, \dot a)(\dot a, a)(b, \dot b)(\dot b, b)(b, \dot b)(a, a)(\dot a, a)(\#, \#)(\#, \#)\dots$$
		The next relevant trace, therefore, contains
		$$(\dot a, \dot a)(b, a)(\dot b, \dot b)(b, b)(a, \dot b)(\dot a, a)(\#, a)(\#, \#)(\#, \#)\dots$$
		Continuing this, the following relevant traces are:
		\begin{flalign*}
		&(\dot b, \dot b)(b, b)(a, \dot b)(\dot a, a)(\#, a)(\#, \#)(\#, \#)\dots\\
		&(\dot a, \dot b)(\#, a)(\#, a)(\#, \#)(\#, \#)\dots\\
		&(\#, \#)(\#, \#)\dots
		\end{flalign*}		
		The relevant traces verify the solution provided on the $\G i$ trace by removing one stone after the other. Thus, the formula is realizable iff the PCP instance has a solution. \qed
	\end{proof}

	
	\subsection{Multiple Universal Trace Quantifiers}
	\label{sec:MultUnivTrace}
	When considering multiple universal trace quantifiers $\forall^*_\pi$, the problem becomes undecidable. This is because in HyperLTL, one can encode distributed architectures -- for which the problem is undecidable -- directly into the formula without using any propositional quantification~\cite{conf/cav/FinkbeinerHLST18}.
	\begin{corollary}
		Realizability of the $\forall^*_\pi$ fragment is in general undecidable.
	\end{corollary}
	%
	However, we show that the realizability problem for formulas with more than one universal trace quantifier is decidable if we restrict ourselves to formulas in the so-called \emph{linear fragment}, i.e., that does not allow an encoding of a distributed architecture. We define the linear fragment of HyperQPTL, where the definitions are adopted from~\cite{conf/cav/FinkbeinerHLST18}.
	
		Let $A,C \subseteq \ap$. We define that atomic propositions $c \in C$ do solely depend on propositions $a \in A$ as the HyperQPTL formula
		\begin{equation*}
		\dep{A}{C}\coloneqq
		\forall \pi \forall \pi' \ldot
		\left(
		\bigvee_{a \in A} (a_\pi \nleftrightarrow a_{\pi'})
		\right)
		\R
		\left(
		\bigwedge_{c \in C} ( c_\pi \leftrightarrow c_{\pi'})
		\right)
		\enspace.
		\end{equation*}
	%
		We define a \emph{collapse} function, which collapses a HyperQPTL formula with a $\forall_\pi^*$ universal quantifier prefix into a formula with a single $\forall_\pi$ quantifier. Propositional quantifiers are preserved by the operation.
		Let $\varphi$ be $ \forall \pi_1 \cdots \forall \pi_n \ldot Q_q^* \ldot \psi$.
		We define the collapsed formula of $\varphi$ as $\collapse(\varphi) \coloneqq \forall \pi \ldot Q_q^* \ldot \psi[\pi_1 \mapsto \pi][\pi_2 \mapsto \pi]\dots[\pi_n \mapsto \pi]$ where $\psi[\pi_i \mapsto \pi]$ replaces all occurrences of $\pi_i$ in $\psi$ with $\pi$.
	
	\begin{lemma}
		Either $\varphi \equiv \collapse(\varphi)$ or $\varphi$ has no equivalent $\forall_\pi^1 \ldot Q_q^*$ formula.
	\end{lemma}
	\begin{proof}
		The collapse function solely works on the trace quantification mechanism of the HyperQPTL formula, by reducing them to a single universal quantification. The theorem has been proven for $\forall^*$ HyperLTL formulas in~\cite{conf/cav/FinkbeinerHLST18}. Inner propositional quantification does not interfere with this mechanism, hence, the proof can be carried out identically. \qed
	\end{proof}

	Now we can formally define the linear $\forall_{\pi}^*$ fragment. Intuitively, we require that every input-output dependency can be ordered linearly, i.e., we are restricted to linear architectures without information forks (see Example~\ref{ex:info_fork}).
	\begin{definition}
		Let $O = \{o_1,\ldots, o_n\}$.
		A HyperQPTL formula $\varphi$ is called \emph{linear} if for all $o_i \in O$ there is a $J_i \subseteq I$ such that $\varphi \land \dep{I}{O} \equiv \collapse(\varphi) \land \bigwedge_{o_i \in O} \dep{J_i}{\set{o_i}}$ and $J_i \subseteq J_{i+1}$ for all $i \leq n$.
	\end{definition}
	This results in the following corollary. Since the universal quantifiers can be collapsed, the resulting problem is the realizability problem of QPTL in a linear architecture, which is decidable~\cite{conf/lics/FinkbeinerS05}.
	
	
	\begin{corollary}
		Realizability of the linear $\forall_\pi^* Q_q^*$ fragment is decidable.
	\end{corollary}

\paragraph{Remark on Complexities.} Our aim was to work out the largest possible fragments for which the realizability problem of HyperQPTL remains decidable. The three fragments for which we could prove decidability all subsume the logic QPTL, for which the realizability problem is known to be non-elementary (already its satisfiability problem is non-elementary~\cite{DBLP:journals/tcs/SistlaVW87}). Hence, realizability of the discussed HyperQPTL fragments has a non-elementary lower bound. Finding interesting fragments for which the problem has a more feasible complexity therefore remains an open challenge.

\section{Experiments}
\label{sec:Experiments}
	\begin{table}[t]
	\centering
	\caption{Experimental results for prompt arbiter.}
	\begin{tabular}{l  l  l  l l}
		\noalign{\smallskip} \hline \noalign{\smallskip}
		instance & bound on system \hspace{.1cm} & bound on $\exists$-strategy \hspace{.1cm}  &  result \hspace{.1cm} & time [sec.] \\
		\noalign{\smallskip} \hline\noalign{\smallskip}
		arbiter-2-prompt & $2$ & $1$ &  unsat & $<1$ \\
		& $2$ & $2$& sat & $<1$ \\
		arbiter-2-full-prompt & $3$ & $1$& unsat & $2.4$ \\
		& $3$ & $2$ & sat &$ 6.0$ \\
		arbiter-3-prompt & $3$& $1$ & unsat & $4.2$ \\
		& $3$ & $2$ & sat & $9.5$ \\
		%
		%
		arbiter-4-prompt & $4$ & $1$ & unsat & $97$ \\
		& $4$ & $2$ & ? & TO \\
	\end{tabular}
	\label{tbl:results}
\end{table}

We have implemented a prototype tool that can solve the HyperQPTL realizability problem using the bounded synthesis approach~\cite{journals/sttt/FinkbeinerS13}.
More concretely, we extended the HyperLTL synthesis tool BoSy~\cite{conf/cav/FaymonvilleFT17,conf/cav/FinkbeinerHLST18,conf/cav/CoenenFST19}. Bosy reduces the HyperLTL synthesis problem to a SMT constraint system which is then solved by $\zthree$~\cite{z3} (for more see~\cite{conf/cav/FinkbeinerHLST18}). We implemented the reduction of HyperQPTL synthesis to HyperLTL synthesis (Corollary~\ref{cor:hyperqptl_2_hyperltl}) in BoSy, such that the tool can also handle HyperQPTL formulas.
We evaluated the tool against a range of benchmarks sets, shown in Table~\ref{tbl:results}. The first column indicates the parameterized benchmark name. The second and third columns indicate the bounds given to the bounded synthesis procedure. The second column is the bound on the size of the system. The newest version of BoSy also bounds the size of the strategy for the existential player, this bound is given in column three. For a detailed explanation of how existential strategies are bounded in BoSy, we refer to~\cite{conf/cav/CoenenFST19}.

We synthesized a range of resource arbiters. Our benchmark set is parametric in the number of clients that can request access to the shared resource (written arbiter-$k$-prompt where $k$ is the number of clients in Table~\ref{tbl:results}). Unlike normal arbiters, we require the arbiter to fulfill promptness for some of the clients, i.e., requests must be answered within a bounded number of steps~\cite{journals/corr/TentrupWZ15}. We state the promptness requirement in HyperQPTL by applying the \emph{alternating-color technique} from~\cite{journals/fmsd/KupfermanPV09}. Intuitively, the alternating-color technique works as follows: We quantify a $q$-sequence that ``changes color" between $q$ and $\neg q$. Each change of color is used as a potential bound. Once a request occurs, the grant must be given withing two changes of color. Thus, the HyperQPTL formulation amounts to the following specifications, here exemplary for 2 clients, where we require promptness only for client 1.
\begin{align}	
	& \label{form:arbiter_mutual}\forall \pi \ldot \G \neg (g^1_\pi \land g^2_\pi) \\
	& \label{form:arbiter_eventually} \forall \pi \ldot \G (r^2_\pi \rightarrow \F g^2_\pi) \\
	& \label{form:arbiter_promptness} \exists q \ldot \forall \pi \ldot \G\F q \land \G \F \neg q \\
	& \nonumber \phantom{\exists q \ldot \forall \pi \ldot} \land \G (r^1_\pi \rightarrow (q \rightarrow (q \LTLuntil (\neg q \LTLuntil g^1_\pi))) \\
	& \nonumber \phantom{\exists q \ldot \forall \pi \ldot \land (r^1_\pi \rightarrow (}  \land (\neg q \rightarrow (\neg q \LTLuntil (q \LTLuntil g^1_\pi)))) \\
	& \label{form:arbiter_spurious}\forall \pi. (\neg g^1_\pi \LTLweakuntil r^1_\pi) \land (\neg g^2_\pi \LTLweakuntil r^2_\pi)
\end{align}
Formula~\ref{form:arbiter_mutual} states mutual exclusion. Formula~\ref{form:arbiter_eventually} states that client 2 must be served eventually (but not within a bounded number of steps). Formula~\ref{form:arbiter_promptness} states the promptness requirement for client 1. It quantifies an alternating $q$-sequence, which serves as a sequence of global bounds that must be respected on all traces $\pi$. Then, if client 1 poses a request, the grant must be given within two changes of the value of $q$. Formula~\ref{form:arbiter_spurious} is only added in benchmarks named arbiter-$k$-full-prompt. It specifies that no spurious grants should be given.

BoSy successfully synthesizes prompt arbiter of up to $3$ states. For a $4$-state prompt arbiter BoSy did not return in reasonable time.

\section{Conclusion}
We studied the hyperlogic HyperQPTL, which combines the concepts of trace relations and $\omega$-regularity. We showed that HyperQPTL is very expressive, it can express properties like \emph{promptness}, \emph{bounded waiting for a grant}, \emph{epistemic} properties, and, in particular, any $\omega$-\emph{regular} property. Those properties are not expressible in previously studied hyperlogics like HyperLTL. At the same time, we argued that the expressiveness of HyperQPTL is optimal in a sense that a more expressive logic for $\omega$-regular hyperproperties would have an undecidable model checking problem.
We furthermore studied the realizability problem of HyperQPTL. 
We showed that realizability is decidable for HyperQPTL fragments that contain properties like promptness.
But still, in contrast to the satisfiability problem, propositional quantification does make the realizability problem of hyperlogics harder. More specifically, the HyperQPTL fragment of formulas with a universal-existential propositional quantifier alternation followed by a single trace quantifier is undecidable in general, even though the projection of the fragment to HyperLTL has a decidable realizability problem.
Lastly, we implemented the bounded synthesis problem for HyperQPTL in the prototype tool BoSy. Using BoSy with HyperQPTL specifications, we have been able to synthesize several resource arbiters.
The synthesis problem of non-linear-time hyperlogics is still open. For example, it is not yet known how to synthesize systems from specifications given in branching-time hyperlogics like HyperCTL$^*$.


%
%
%
\bibliographystyle{splncs04}
\bibliography{main.bib}
\appendix
\end{document}